\documentclass[conference]{IEEEtran}
\ifCLASSINFOpdf
\else
\fi
\usepackage{cite}
 \usepackage{color}
 \usepackage{graphicx}
 \usepackage{array}
 \usepackage{amsmath}
\usepackage{url}
\usepackage{amsmath}
\newtheorem{theorem}{Theorem}
\newtheorem{proof}{Proof}[section]
\IEEEoverridecommandlockouts

\begin{document}
%
\title{Cloud based Real-Time and Low Latency Scientific Event Analysis}
\author{\IEEEauthorblockN{Chen Yang\IEEEauthorrefmark{1},
Xiaofeng Meng\IEEEauthorrefmark{1}\thanks{Corresponding author by Xiaofeng Meng},
Zhihui Du\IEEEauthorrefmark{2}
}
\IEEEauthorblockA{\IEEEauthorrefmark{1}School of Information, Renmin University, Beijing, China\\ Email: \{yang\_chen, xfmeng\}@ruc.edu.cn}
\IEEEauthorblockA{\IEEEauthorrefmark{2}Department of Computer Science and Technology, Tsinghua University, China\\
Email: duzh@tsinghua.edu.cn}
}
\maketitle

\begin{abstract}
Astronomy is well recognized as big data driven science. As the novel observation infrastructures are developed, the sky survey cycles have been shortened from a few days to a few seconds, causing data processing pressure to shift from offline to online. However, existing scientific databases focus on offline analysis of long-term historical data, not real-time and low latency analysis of large-scale newly arriving data.

In this paper, a cloud based method is proposed to efficiently analyze scientific events on large-scale newly arriving data. The solution is implemented as a highly efficient system, namely Aserv. A set of compact data store and index structures are proposed to describe the proposed scientific events and a typical analysis pattern is formulized as a set of query operations. Domain aware filter, accuracy aware data partition, highly efficient index and frequently used statistical data designs are four key methods to optimize the performance of Aserv. Experimental results under the typical cloud environment show that the presented optimization mechanism can meet the low latency demand for both large data insertion and scientific event analysis. Aserv can insert 3.5 million rows of data within 3 seconds and perform the heaviest query on 6.7 billion rows of data also within 3 seconds. Furthermore, a performance model is given to help Aserv choose the right cloud resource setup to meet the guaranteed real-time performance requirement.
\end{abstract}

\section{Introduction}
In astronomy, short astronomical phenomena mean grand scientific discoveries. Up to now, only 10 astronomical phenomena lasting within 1 day are found. Actually, existing astronomy projects cannot effectively search optical transient sources who are during a few hours, due to the long sky survey cycles. For example, the survey cycle of both SDSS\cite{szalay2002sdss} and LSST\cite{Wang2011Qserv} are 3-5 days. For searching short and unknown astronomical phenomena, the fast sky survey projects have become a new trend. For example, GWAC (Ground-based Wide Angle Camera)\cite{wan2016column}, having the shortest sky survey cycle in the world, continuously observes the fixed 1/4 Northern Hemisphere within 15 seconds. The new instruments lead a new kind of big scientific data and different analysis needs.


The new survey data provides scientists a completely new way to achieve scientific discovery. Scientists often need to launch an analytical query on newly arriving data to confirm a scientific event and issue an alert as soon as possible. This requirement is basic since short astronomical phenomena are transient and hard to reproduce, such as microlensing. Thus, it tends to the online analysis and the analysis methods are totally different from offline analysis. We extract three typical analysis methods: probing, listing and stretching. They formulize the analysis behavior of scientists on newly arriving data from the general view to the deep insight.


\begin{figure}[t]
  \centering
  \includegraphics[width=2.5in,height=1.7in]{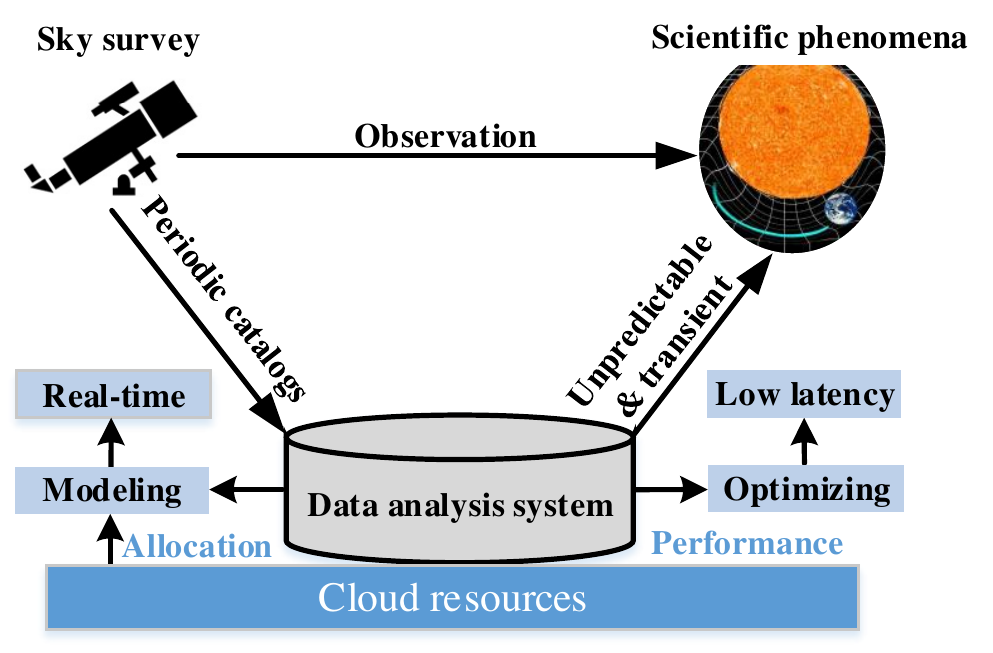}
  \caption{Real-time and low latency scientific event analysis on cloud}\label{fig:problemExample}
\end{figure}

To support the wanted analysis methods on cloud, data analysis system behind these instruments desires both \textbf{real-time and low latency} ability, as shown in Figure \ref{fig:problemExample}. As the survey data is periodically collected from scientific objects, data analysis system must take real-time processing to guarantee the data temporal consistency\cite{Lindstr2008Real} and no data loss. Since scientific events discussed here are unpredictable and transient, low latency data analysis is necessary for scientists to identify the events as early as possible. In addition, the expense of large scientific projects often exceeds the budget because the projects often last long and more time and many unpredicted difficulties will happen. So suitable cloud resource setup is necessary to reduce the fixed expense of data analysis system.

 To recap, on the premise of the low resource overhead, the data insertion and query time must be less than the survey cycle and the lower it is, the better the performance is. For the data insertion operation, challenges mainly involve: (1) the cost of distributed processing, (2) data size and (3) the latency trade-off between insertion and query. The large-scale data insertion will take up much network and storage resources. Compression will incur more computational cost and sampling may lose some key data that is often unacceptable for scientists when using these techniques to reduce the data size. If we simply insert data collected at a survey cycle as a catalog file, the query latency on unstructured data will be very high. To enable low latency query, index on scientific events is necessary. However, the expensive cost in index update and insertion often prevents the low latency query.

 Based on the above, we propose Aserv, a lightweight system for real-time and low latency scientific event analysis. The key idea to improve the performance is \textbf{cutting down unnecessary cost as much as possible}. Without losing the availability of our system, three policies are developed to improve the overall performance: (1) removing irrelevant data; (2) adjusting the query accuracy to an acceptable range instead of as high as possible and (3) eliminating too expensive operations. Furthermore, we also develop a \textbf{performance model} to help Aserv determine resource setup to meet the performance constraint.

A set of exquisite optimization methods are employed on the two major components of Aserv: (1) the data insertion part is a real-time processing pipeline to ingest scientific data and load them into key-value store, and (2) the query engine supports low latency scientific event analysis. The insertion component includes three major modules: filter, data organization and pre-analysis. We only select highly related information from original data stream to achieve a significant data reduction that can save both computation and storage cost. Data organization module physically partitions scientific data into different sections so Aserv can greatly reduce network requirement and improve the insertion performance by ingesting partition data, instead of independent data tuples. Correspondingly, query engine in Aserv also implements an accuracy aware search strategy to improve the query performance. In addition, Aserv builds a highly effective index in data organization module and produces statistical data for scientific events in pre-analysis module. Both of them can avoid the access of original data and they also have the insertion-friendly structure.


We evaluate Aserv in a real astronomical project\cite{wan2016column}. Experimental results show that Aserv can really work. In summary, the major contributions are as follows:
\begin{itemize}
\item We propose the real-time and low latency analysis problem in fast sky survey and formulize it as three typical query operations.
\item We develop a cloud based distributed system Aserv for real-time and low latency scientific event analysis. Aserv employs several efficient policies to improve the system performance, including a filter strategy DAfilter, a partitioner EPgrid, SEPI index and an approach PCAG to generate frequently used statistical data.
\item We present a performance model for Aserv. It can help Aserv meet performance constraints under cloud scenario.
\end{itemize}

The rest of the paper is organized as follows. The problem description is in Section \ref{section:dataModel}. Aserv framework is in Section \ref{section:Aserv}. The performance model is in Section \ref{section:Performance}. Our experimental results are in Section \ref{section:experiments}. The summary is in Section \ref{section:summary}.

\section{Scientific Event Analysis}\label{section:dataModel}
In this paper, we focus on the fast sky survey which has been very popular in astronomy recently. Here, the observation instruments consist of multiple observation units, which continuously observe a fixed region per survey cycle $ct$. We assume that the fixed region area is $es$ in Euclidean 2-space and each unit deals with a square sub-area $s$ in $es$. This hypothesis makes sense in many cases. For example, observation instruments collect data by taking images, so that Euclidean distance is usually used to distinguish objects in astronomy\cite{webpage:Topcat}. Noting that observation units in our assumption are logical. In the extreme case, instruments may only have one physical observation unit, but we can still partition the observed area to simulate multiple logical observation units. Thus, it does not affect our assumption.
\begin{figure}[t]
  \centering
  \includegraphics[width=3.5in,height=1.5in]{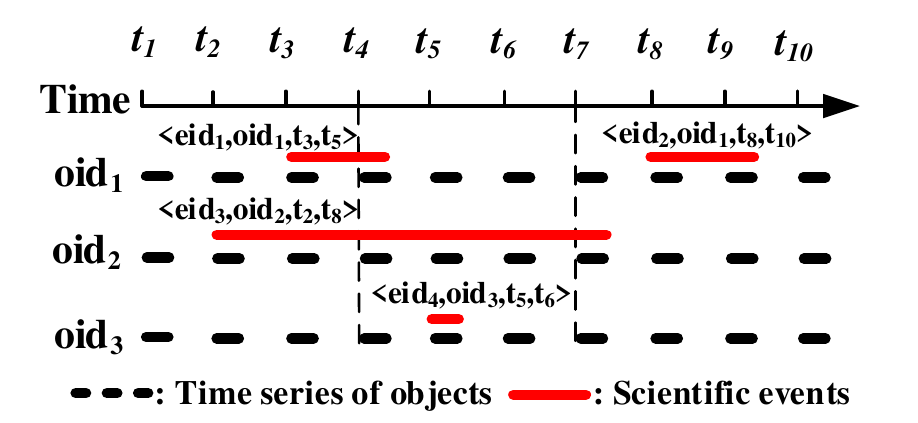}
  \caption{Example of typical analysis for scientific events}\label{fig:Examples}
\end{figure}

\textbf{Data model}. During the observation (i.e., one night), we assume that the instruments can observe $n$ objects. When there are enough objects in the observation region, the resolution limitation of observation units causes the collected objects to be almost constant, such as $\sim$175,600 objects for GWAC\cite{wan2016column}. We assume that the observed objects are well distributed in space. For each cycle, the data collected by each observation unit is organized as the catalog file with the same size. Data tuples in catalogs are in the form of $<oid,x,y,t,d_1,...,d_m>$, where $x$ and $y$ represent locations in $es$. $Oid$ and $t$ are the object name and timestamp, respectively. The rest are data items. Tuples of different objects in the same catalog have the same timestamp. Along the time dimension, tuples of the same object in different catalogs form an amount of time series data.

\textbf{Scientific event}. In addition, each observation unit contains an event detector which can recognize objects who may subject to scientific phenomena from each catalog. Finally, emit a scientific event set $Eset$ including the candidate object $oid$s. Thus, we define a scientific event in the form of $<eid,oid,stime,etime>$ where $eid$, $stime$ and $etime$ are the scientific event ID, event start and end time, respectively. This definition suggests that an object may appear multiple scientific phenomena during the observation.

    \begin{figure*}[t]
      \centering
      \includegraphics[width=5in,height=2.0in]{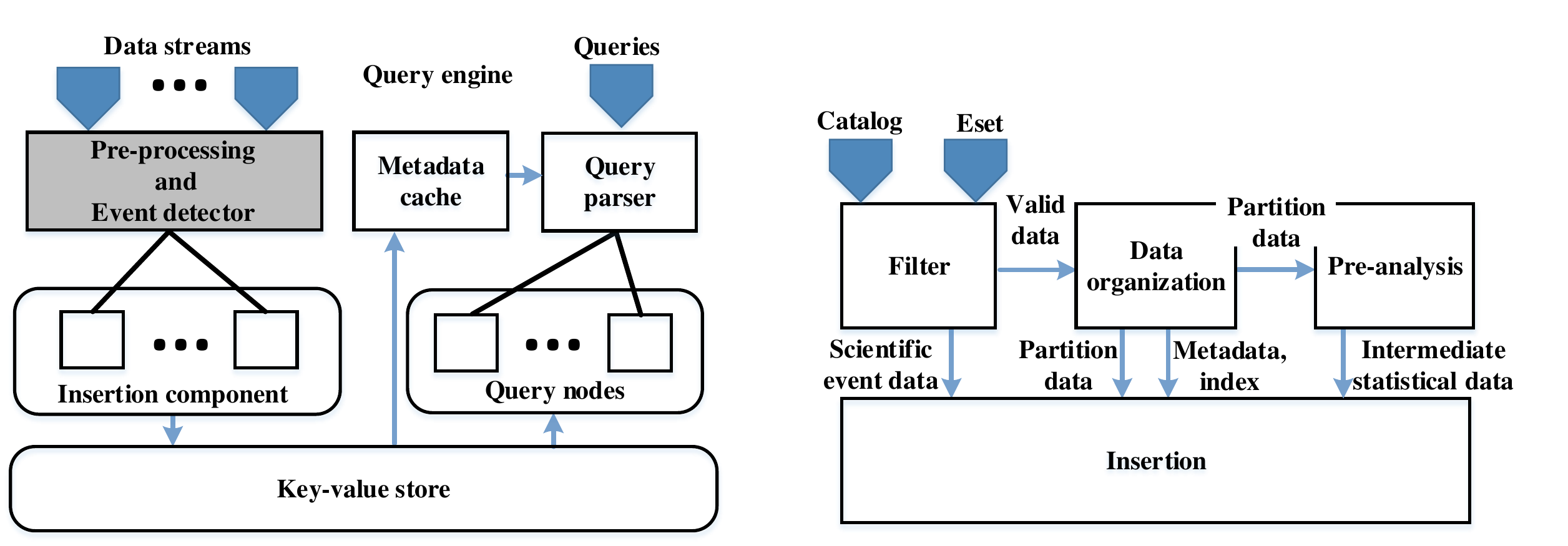}
      \caption{Aserv (Left) includes insertion component, query engine and key-value store. The insertion component (Right) follows ``master/slave" mode. The master is exploited to register and monitor the workers, and each worker node works for an observation unit. Workers ingest the catalog file and $Eset$ from pre-processing module and event detector and load them into key-value store. Query engine supports three typical analysis methods cited above.}\label{fig:Aserv}
    \end{figure*}
For a given query, Aserv must handle two basic operators: $region(x,y,r)$ and $timeinterval(t_s,t_e)$. The space constraint operator $region(x,y,r)$ means to search scientific events within a circle where $x$ and $y$ are the center locations and $r$ is the radius. $Region(x,y,r)$ is suitable for the neighborhood search. In some cases, the search region may be the rectangle. Relatively speaking, the circle region is more commonly used and more complex to implement. Thus, we do not discuss the rectangle region. $Timeinterval(t_s,t_e)$ is a time interval between $t_s$ and $t_e$.

In real-time scientific discovery, we propose three analysis methods for scientific events. An example is shown in Figure \ref{fig:Examples}, where $oid_{1-3}$ are observed by the same observation unit from $t_1$ to $t_{10}$. They formulize the analysis behavior of scientists.

\textbf{Probing analysis.} It mainly returns overview information that can help scientists to have a quick view on scientific events, such as some aggregation operations. This analytical method is useful in the alert of scientific events. In this paper, we focus on scientific event count, which answers how many scientific events in $timeinterval(t_s,t_e)$, because scientists are only interested in the occurrence of scientific events. For example, it returns 0 in $timeinterval(t_1,t_2)$ meaning no scientific events and 3 in $timeinterval(t_4,t_7)$. It, as the most frequent query, must have a low latency.

\textbf{Listing analysis.} It returns the complete information of scientific events. When scientists find the alert by probing analysis, they can use listing analysis to return the complete time series of scientific events in $timeinterval(t_s,t_e)$. We use interval query\cite{salzberg1999comparison} to implement listing analysis. Assuming that data items in a data set are time-evolving, as long as the data items appear within the given time interval, interval query will return all the items in this set. For example, it returns the time series corresponding to $eid_1$, $eid_3$ and $eid_4$ in $timeinterval(t_4,t_7)$, i.e., for scientific event $eid_1$ returning the time series of $oid_1$ between $t_3$ and $t_5$, etc. Although the durations of scientific events may only intersect the given time interval, we still return the complete time series, because scientists are always more concerned with a complete change in a scientific phenomenon.

\textbf{Stretching analysis.} It mainly returns the extend information of any scientific event. Given a scientific event found by listing analysis, scientists might be interested in its surroundings, such as a larger range of time or space. Therefore, it is a complement to listing analysis. We employ temporal range query to implement the time stretch, because scientists through this query could have a deeper insight on what happens before and after a scientific event. For a scientific event, this query returns a time series range in $timeinterval(stime-\Delta t_1,etime+\Delta t_2)$. As an example, for $eid_4$ this query returns the time series range of $oid_3$ between $t_4$ and $t_7$ using $\Delta t_1=\Delta t_2=1$.

Probing analysis and listing analysis can run with a space restriction. For example, scientists could perform listing analysis with both $region(x,y,r)$ and $timeinterval(t_s,t_e)$.

Aserv needs to meet the following performance requirements: (1) for each data survey, the insertion latency is required to be less than $ct$, and (2) the query latency over data tuples collected during the observation is required to be less than $ct$.

\section{Aserv Framework}\label{section:Aserv}
Aserv includes two major components and a key-value cloud store, as shown in Figure \ref{fig:Aserv}. The pre-processing module mainly includes the necessary scientific precessing, such as cross-match in astronomy\cite{nieto2007cross}. The event detector is used to search $Eset$ in each catalog. The pre-processing module and event detector have been implemented successfully in existing work\cite{Xu2013A, Feng2017Real}, so here we focus on the unsolved parts. The major modules of Aserv are discussed as follows.

 \textbf{Filter}. For the catalogs, the data dimension $m$ is usually large, such as 25 columns per data tuple collected by GWAC. It is important to reduce the data size, especially for in-memory store. Actually, in our online analysis scenario scientists only focus on scientific events and several major attributes. To significantly improve the system performance, we must filter out unnecessary information from complete data using the domain knowledge\cite{Klasky2017Exacution} and simultaneously achieve the efficient data reduction. We develop a \textbf{Domain Aware Filter (or DAfilter)} to support our optimization.
 \begin{itemize}
   \item In the first step, for every data tuple we filter the $c$ major attributes and store $<oid,x,y,t,d_1,...,d_c>$ where $c<<m$ (e.g., $c=1$ for GWAC) as \textbf{valid data}, instead of the original data tuple.
   \item In the second step, for each scientific event we create an $eid$ as a key and model its time series as key-list structure. In detail, for a given time $T$, we select the object $oid \in Eset$, and additionally append its original data tuple (i.e., $< oid,x,y,t,d_1,...,d_m>$) into key-list structure as \textbf{scientific event data}, because the complete information is also useful for scientists to analyze scientific events.
 \end{itemize}

 \textbf{Data organization}. In this module, we generate the metadta, partition the valid data and build the index. At the beginning, we use EPgrid (Section \ref{section:EPgrid}) to partition the observed region $es$ once to generate partition \textbf{metadata}. For key-value stores, the number of keys significantly impacts the insertion performance\cite{atikoglu2012workload}, so that we also partition the valid data into \textbf{partition data} to reduce keys. We organize the partition data with the same ID into key-list structure along the time dimension. We do not partition scientific event data due to less numbers, but we still record their partition IDs into scientific event data for easy querying. In addition, we use $Eset$ to update SEPI index (Section \ref{section:SEPI}).

 \textbf{Pre-analysis}. Probing analysis is important for scientific event alert. However, it has a poor performance to scan scientific event data or index. Thus, this module receives partition data and $Eset$ from data organization module and generates the \textbf{intermediate statistical data} to speed up frequently used probing analysis (Section \ref{section:ScientificEventCount}).

 \textbf{Insertion and key-value store}. The data, produced by the aforementioned modules, will be transformed into key-value pairs and ingest them into the key-value cloud store. On the one hand, the linear scalability of key-value stores is helpful to meet the performance constraints\cite{Stonebraker201110}. On the other hand, scientific event data can be well described as key-list structure, natively supported by key-value store.

 \textbf{Query engine}. Aserv in advance loads partition metadata from key-value store and caches them in memory to speed up the region search. When a query request comes, Aserv will approximatively parse $region$ operator to read the right partition data for stretching analysis, intermediate statistical data for probing analysis or scientific event data for listing analysis. Finally, the invoked query will run on them. We introduce these techniques as follows.
\subsection{Accuracy Aware Data Partition}\label{section:EPgrid}
In this section, we employ the grid scheme to partition the observation region of each observation unit. When $region$ operator is parsed, query engine will analyze partition metadata in local memory and filter partitions who are covered by search circle, no matter how much the covered area is. Finally, return partition IDs to index objects in the corresponding partitions. As shown in Figure \ref{fig:SEPIExample}, four partitions can be found.
Obviously, the total area covered by our strategy is always greater than the given search circle. Although scientists can tolerate that the irrelevant area is covered, it must be within a reasonable range. Thus, it is important to make sense of relation between the region search accuracy and the number of partitions.
\begin{figure}[t]
  \centering
  \includegraphics[width=3.5in,height=1.3in]{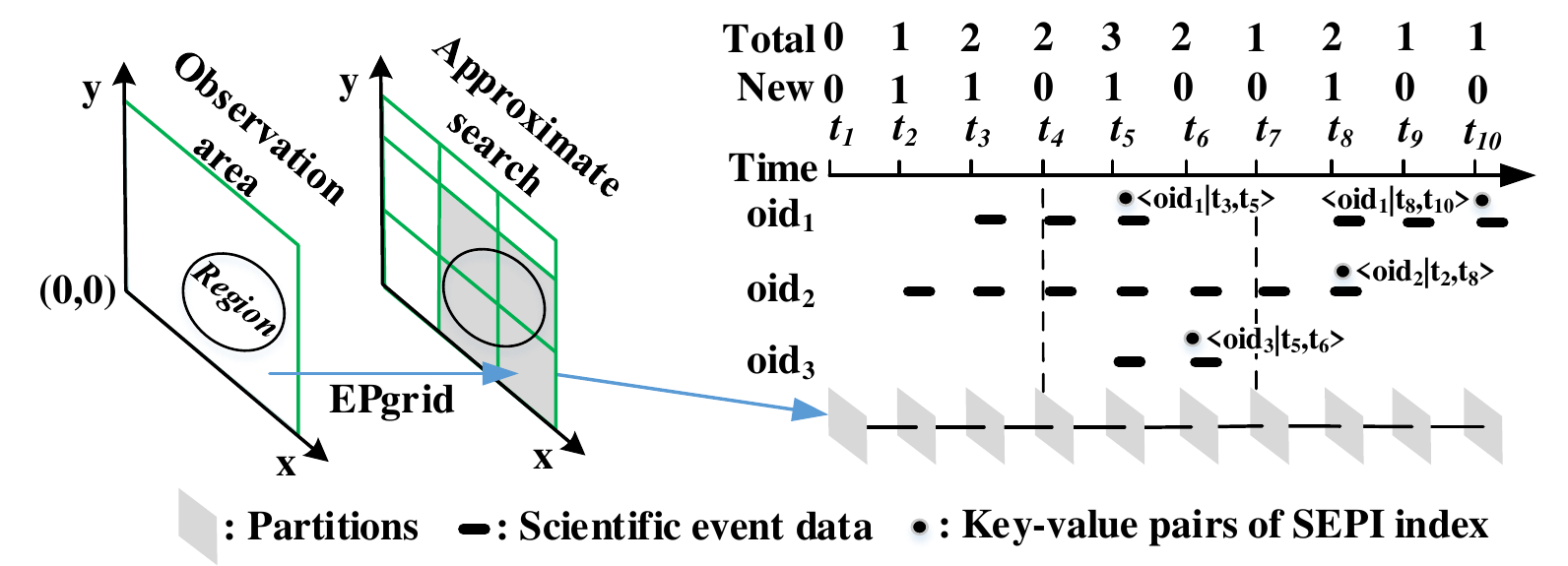}
  \caption{Example of partition and approximate search (Left), and the data storage scheme of one partition which includes $oid_{1-3}$ (Right).}\label{fig:SEPIExample}
\end{figure}

\begin{theorem}
  Given the acceptable minimum accuracy $\alpha$, the search radius $r$ and the observation sub-area $s$ of an observation unit, when the grid size is equal to each other, the minimum grid number $gn$ satisfies the equality defined as
\begin{equation}\label{eq:accRegion}
gn = \left\lceil {\frac{{64s{\alpha ^2}}}{{{{(\pi r(1 - \alpha ))}^2}}} } \right\rceil,
\end{equation}
which can ensure that the region search accuracy with radius $r$ is not lower than $\alpha$.
\end{theorem}

\begin{proof}
We first define $l$ and $w$ as the length and width for each grid, respectively. Further, the total area of grids covered by the approximate strategy is defined as $gs$. The area of grids which are covered by the search circle boundary is less than the area of grids which are covered by the circumscribed square boundary. Thus, $gs$ satisfies the inequality defined as
\begin{equation*}
  gs \le 4r(l + w) + \pi {r^2},
\end{equation*}
where $4r(l + w)$ is the area of grids which are covered by the circumscribed square boundary. We have known $gs \ge \pi {r^2}$ which is the size of search circle. Therefore, the region search accuracy $acc(region)=\pi {r^2}/gs$ satisfies the following inequality defined as
\begin{equation*}
acc(region) \ge \frac{{\pi r}}{{4(l + w) + \pi r}}.
\end{equation*}
It is consistent with the intuition that the search radius is positively correlated with the region search accuracy when grids are fixed. The acceptable minimum accuracy is known as $\alpha$, so we can solve the inequality as
\begin{equation*}
{l + w} \le \frac{{\pi{r}(1 - \alpha )}}{{4\alpha }}.
\end{equation*}
For obtaining the minimum grid number, we set $l=w$ and then $l=\pi r(1 - \alpha )/8\alpha$. Thus, solve $gn=s/l^2$ as Eq. (\ref{eq:accRegion}), and the minimum grid number increases with the reduction of the search radius. When giving the acceptable minimum radius, we can solve the final grid number.
\end{proof}

The occurrence of scientific events is usually independent to each other and almost can be seen as the uniform distribution. Therefore, the searched scientific events only depend on the search area, causing $acc(region)$ to be actually equal to the query accuracy. When Eq. (\ref{eq:accRegion}) is met, the query accuracy should be also greater than $\alpha$.

According to the suggestion of Eq. (\ref{eq:accRegion}), we implement EPgrid as an even grid scheme to partition the sub-area of each observation unit. For the given number of partitions, EPgrid tries to ensure that each grid is a square. Then, give every grid a unique partition ID. In addition, we record the lower left and upper right positions as partition metadata. Therefore, EPgrid can determine the partition ID of an object by a simple hash function, omitted due to space constraints. In addition, we design a map-only job to parse $region$, which only filters partitions and no data is transmitted over the network. Therefore, the performance to parse $region$ is easily improved with the extension of the cluster scale.


\subsection{SEPI Index}\label{section:SEPI}
SEPI index is employed to support listing analysis. The listing analysis searches $eid$s of scientific events using SEPI and loads the time-series of the corresponding scientific events. SEPI can be efficiently inserted and updated, and the index tier supports the efficiently distributed query. The most similar index to SEPI is EPI\cite{sfakianakis2013interval}. Compared with EPI, indexing on SEPI is faster by eliminating expensive operations, and its size is half of EPI.

If an object $oid$ just appears in $Eset$ at the current time $t$, it means a new scientific event. We set $eid=oid|t$ where ``$|$" means the concatenation of two strings, and emit a key-value pair $<eid,t>$ to key-value store. Otherwise, we update the value of existing $eid$ into $t$. We will keep updating the SEPI index until scientific events end. In other words, in key-value store we only keep the key-value pairs $<oid|stime, etime>$. Because we only store etime as the value, we call the index as Single Endpoint Index (or SEPI). SEPI is actually a set of key-value pairs, so that we can insert and update items very fast.

Giving a $timeinterval(t_s,t_e)$ to listing analysis, the query using SEPI is as follows. First, we execute one scan in parallel on SEPI inside key-value store: a $scan()$ for all key-value pairs of SEPI with values in $[t_s,+\infty)$ and loading them into the heap space of query engine. Second, we also execute one filter in parallel on key-value pairs loaded out: a $filter()$ for key-value pairs with $stime \le t_e$ where stime can be extracted from the key string (i.e., $oid|stime$). After the filter, scientific events whose time intervals intersect $timeinterval(t_s,t_e)$ will be found. As shown in Figure \ref{fig:SEPIExample}, in $timeinterval(t_4,t_7)$ three scientific events can be found.

More specifically, even though $t_s$ is little causing a large part of SEPI to be spanned, it does not have a dramatic impact on the performance. First, Aserv is designed for real-time analysis of scientific events during the observation in which data size is large but the observation duration is not long. Second, we directly use the key-value store's $scan()$ primitive to load data, not resulting in more overhead. Finally, query engine only scans SEPI once and the filter operation is a may-only job which is inherently fast on cloud systems.

EPI\cite{sfakianakis2013interval} keeps two key-value pairs for each scientific event in key-value store, where one records start endpoint and another records end endpoint. Given a $timeinterval(t_s,t_e)$, two scans need to be executed. One scan is to find scientific events whose endpoints appear in $[t_s,t_e]$. In addition, for scientific event whose both endpoints are contained in $[t_s,t_e]$, $distinct()$ operation is performed to remove duplicates. Another scan to find Scientific events who pierce $[t_s,t_e]$. Finally, return the union of the two result sets. Actually, SEPI, compared with EPI, excludes $distinct()$ and an extra scan operation. In addition, SEPI's size is only half of EPI, since one key-value pair is kept for each scientific event.

\subsection{Partition Count Aggregation}\label{section:ScientificEventCount}
Partition count aggregation (or PCAG) can support efficient probing analysis. The main idea is that we in advance generate the count for each partition. These counts will be merged to solve the final result to avoid the scan of original data and index when launching probing analysis.

For a given time $T$, we partition $Eset$ using EPgrid with the same parameters. In every partition, we count the total number $Total$ of $oid$s in $Eset$ and the number $New$ of new scientific events whose $stime=T$. For example, as shown in Figure \ref{fig:SEPIExample} $Total=2$ and $New=1$ at $t_3$ mean that two scientific events run through $t_3$ and one of them is just emerging. We generate a key-value pair as \textbf{Intermediate Count Result (or ICR)} where the key contains the partition ID and the value is $Total|New|T$. The new ICR is appended into key-value store. ICRs of the same partition are organized as a key-list structure.

\begin{equation}\label{eq:count}
count(p) = Total({t_s}) + \sum\limits_{i = {t_{s + 1}}}^{{t_e}} {New(i)}
\end{equation}

Query engine first parses $region$ and $timeinterval(t_s,t_e)$ operators to obtain the partition IDs and the time range, and then loads the corresponding ICRs. For a partition $p$, probing analysis satisfies Eq. (\ref{eq:count}). For example, as shown in Figure \ref{fig:SEPIExample} in $timeinterval(t_4,t_7)$ we first load one partition ICRs between $t_4$ and $t_7$. $Total(t_4)$ is 2 and the sum of $New$ between $t_5$ and $t_7$ is 1, so that the count is 3. The final count is the sum of counts of all searched partitions. For cloud systems, our method is also easy to implement. For each partition, we employ one map task to process one partition ICRs. Finally, employ one reduce task to add up counts of different partitions.

Obviously, ICRs will impact not only the Aserv's insertion performance but also the accuracy of probing analysis due to the approximate search strategy of $region$ operator. However, the number of ICRs at $T$ is equal to the number of partitions. An acceptable query accuracy can be got through adjusting the number of partitions using Eq. (\ref{eq:accRegion}).

\section{Performance Constraint}\label{section:Performance}
Aserv can meet the performance constraints in two ways. First, we try to design map-only jobs to enable the predictable performance. All tasks in map-only jobs only process local data, so that the scale overhead is very little. For example, both parsing $region$ and scanning SEPI are map-only jobs. In addition, probing analysis only has a reduce task. Actually, the insertion component in essence can be treated as a map-only job, because workers have no data transmission to each other over the network. Second, on cloud we can scale out the cluster size to adjust the Aserv's performance. Therefore, consider a problem whether the performance constraints can be met for a given cluster size $K$.

The insertion latency consists of two parts: processing time and storage time spent to ingest data into key-value store. Aserv's insertion component can be seen as load-balance due to the same catalog size. Due to the map-only feature, the processing time is equal to the time $f_p(V_n/K)$ spent on processing $V_n/K$ data by one worker, where $V_n$ is the data size of $n$ objects collected per cycle. Using a similar derivation, the storage time is nearly equal to $f_s(V_s/K)$ due to the linear scalability of key-value cloud store\cite{Stonebraker201110} where $V_s$ is the size of the stored data. Then, we define the insertion latency with the cluster size $K$ as
\begin{equation}\label{eq:insertionlatency}
{f_p}({V_n}/K) + {f_s}({V_s}/K) \le ct,
\end{equation}
where $ct$ is the survey cycle. In addition, $f_p+f_s$ is easy to be estimated. We only run the insertion component using one worker with the data size $V_n/K$ and capture the actual insertion latency as $f_p+f_s$.

Similarly, the query latency also involves the reading time $f_r(V_r/K)$ spent on loading data from key-value store and the query execution time $f_q(V_r/K)+f_o(K)$, defined as
 \begin{equation}\label{eq:querylatency}
   f_r(V_r/K)+f_q(V_r/K)+f_o(K)\le ct.
 \end{equation}
More specifically, $f_r$ and $f_q$ are the parallel time being like to $f_p$, but the distributed workloads will incur the scale overhead $f_o$. We find that query workloads in Aserv only exchange data over the cluster by the shuffle pattern, so that $f_o$ includes shuffle I/O overhead and some fixed overhead, such as setting up processes or time spent in serial computation, etc. Actually, we can assign $f_o={\theta}_1K+{\theta}_2$ where ${\theta}_1$ and ${\theta}_2$ are the constants, because shuffle phases satisfy all-to-one communication pattern\cite{venkataraman2016ernest}. For estimating ${\theta}_1$ and ${\theta}_2$, we can set a cluster size $K'$ ($K'<K$) and run the query workload over the data size $K'V_r/K$ to capture the execution time $f_a(K'V_r/K)$. Therefore, we can set $f_o(K')=f_a(K'V_r/K)-f_r(V_r/K)-f_q(V_r/K)$. We also assign $K'$ with different values to solve $f_o$ as the training data, so that ${\theta}_1$ and ${\theta}_2$ can be solved by the linear regression. We consider that Aserv is to nearly meet the performance constraints, when $K$ satisfies both Eq. (\ref{eq:insertionlatency}) and (\ref{eq:querylatency}).

In essence, our approach is to measure the parallel time and predict the communication time. It has important reference value for predicting the performance of short-running tasks. Many models, such as Ernest\cite{venkataraman2016ernest}, are effective to fit the performance of long-running tasks due to the clear computation pattern. However, the computation pattern in short-running tasks is hard to be fitted due to the strong noise. Thus, our approach is  more suitable for short-running tasks. The strategy how to measure Aserv's parallel time is shown in Section \ref{section:performanceConstraintEvaluation}.

\section{Experiments}\label{section:experiments}
We evaluate Aserv from four views: insertion latency, data reduction rate, query latency and query accuracy under a typical astronomical scenario GWAC\cite{wan2016column} in which each observation unit can collect $\sim$175,600 objects per 15 seconds and one observation lasts 8 hours (about 1,920 time points). We simulate GWAC with a data generator.

\textbf{Data generator}. Our data generator follows the ``master/slave" mode, where each sub-generator simulates an observation unit. A sub-generator produces a catalog file per cycle, including $\sim$175,600 lines and 25 columns. The object locations are referenced from the standard UCAC4 catalog\cite{zacharias2013fourth}. In addition, we simulate scientific event signals by setting the $Eset$ size to subject to the geometric distribution and the locations of scientific events to subject to the uniform distribution. The duration of each scientific event is random.

 \textbf{Cluster setup}. We take our experiment on 20 cloud instances supported by Computer Network Information Center, Chinese Academy of Sciences. Each instance has 12 CPU cores (1.6 GHz per core) and 96 GB RAM. The network bandwidth is 10 Gbps. For data generator, we launch 19 sub-generators (i.e., one per instance) and use the last instance as the master. Finally, our cluster will generate 3.5 million rows of catalog data per 15 seconds, and 6.7 billion rows of catalog data per 8 hours. We build the Aserv's cluster on the same 20 instances where each worker in insertion component loads catalogs produced by a sub-generator on the local machine. We use C++ to implement the insertion component and exploit Redis cluster 3.2.11\cite{webpage:Redis} as the storage system. Spark 1.6.3\cite{webpage:Spark} is used for query processing. Experiments consist of three parts as follows.
  \begin{itemize}
    \item We compare insertion latency and data reduction rate under three different numbers of partitions. The insertion latency is 2.35 seconds and Aserv can achieve 2.23x data reduction rate under the optimal number of partitions (i.e., 10,000).
    \item We show the performance of three analysis methods and demonstrate query accuracy. They in Aserv can satisfy the interactive performance. Probing analysis using PCAG is 1.57x-2.28x faster than the existing implementations. Listing analysis by SEPI is 2.22x faster than it by EPI. Stretching analysis can respond in milliseconds. The query accuracy is achieved to 0.9.
    \item We use a few machines to predict Aserv's performance over the larger cluster size. Our performance model is effective to Aserv. The accuracy of predicted insertions latency is 0.96, and it is 0.86 for query latency.
  \end{itemize}


\begin{figure*}[t]
  \begin{minipage}[t]{0.33\linewidth} 
    \centering
    \includegraphics[width=2.2in,height=1.5in]{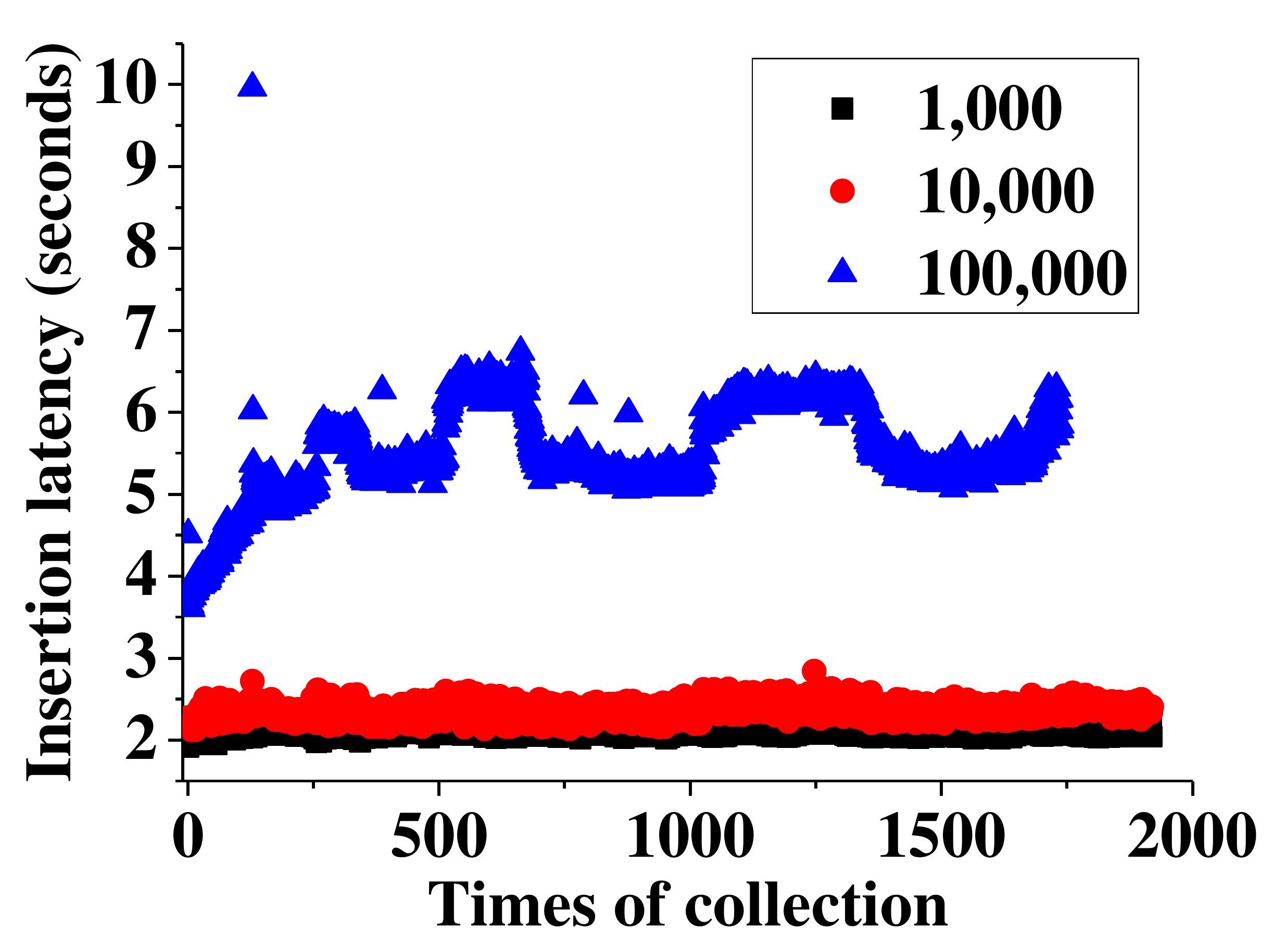}
    \caption{Insertion latency}
    \label{fig:Insertionlatency}
  \end{minipage}%
  \begin{minipage}[t]{0.33\linewidth}
    \centering
    \includegraphics[width=2in]{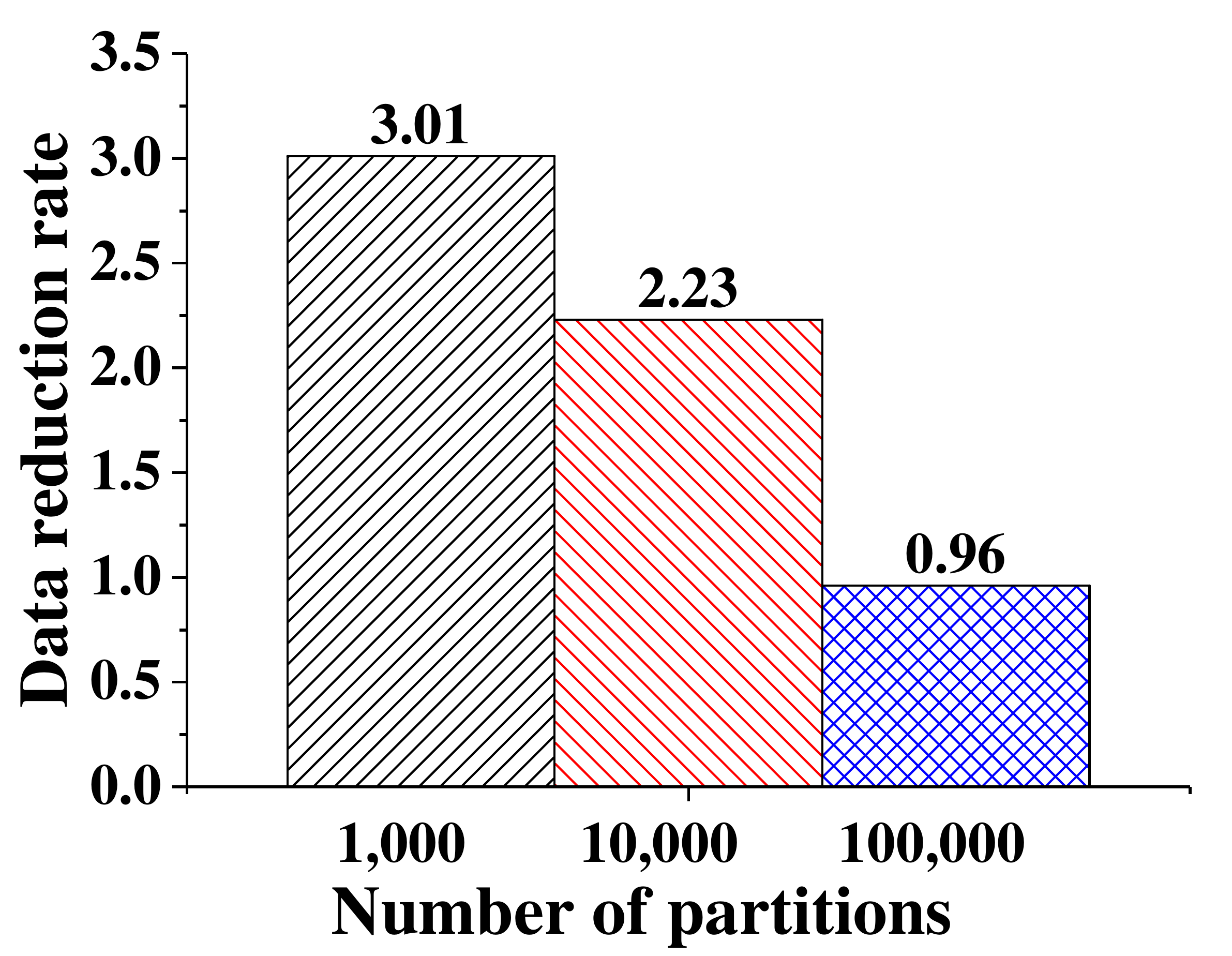}
    \caption{Data reduction rate}
    \label{fig:datareductionratio}
  \end{minipage}
  \begin{minipage}[t]{0.33\linewidth}
    \centering
    \includegraphics[width=2in]{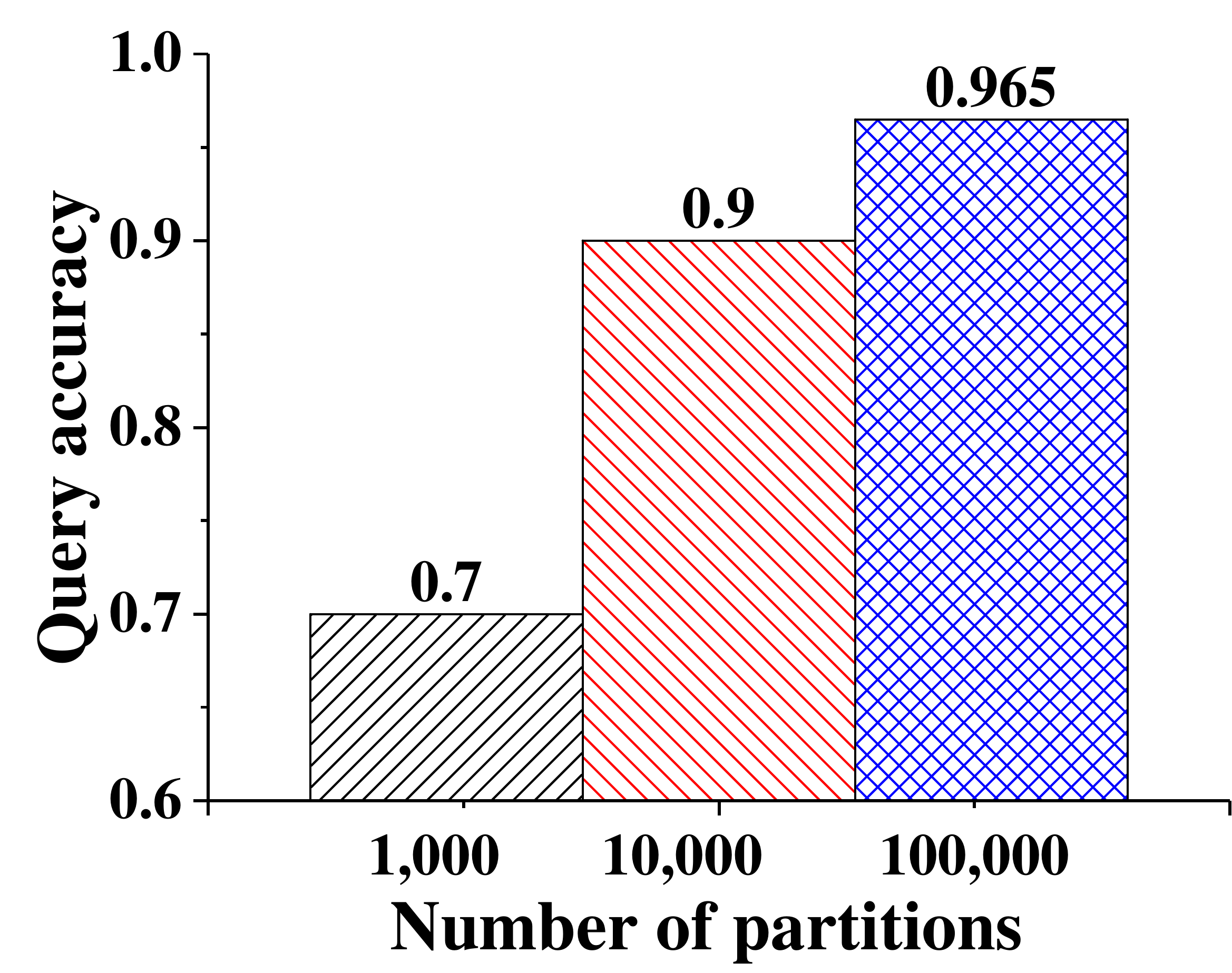}
    \caption{Query accuracy}
    \label{fig:queryaccuracy}
  \end{minipage}
\end{figure*}

\subsection{Insertion Performance Evaluation}\label{section:ClusterSize}
In general, the minimum region search accuracy $\alpha$, which can be accepted by scientists, is 0.8.
The minimum search circle is about 3\% of the area observed by an observation unit. Based on these conditions, the available number of partitions is solved to be about 10,000 for every observation unit using Eq. (\ref{eq:accRegion}). In the experiment, we do not only show the Aserv's performance under 10,000 partitions, but also demonstrate it under both 1,000 and 100,000 partitions as the comparisons.

As shown in Figure \ref{fig:Insertionlatency}, at the first two cases Aserv can finish 1,920 times of data collection. However, we only collect data 1,738 times under 100,000 partitions because too many keys (i.e., partitions) can cause Redis cluster to fail frequently as the survey data continues to be ingested.

\textbf{Insertion latency}. At the three cases, the insertion performance constraint can be met. However, as the number of partitions increases, the insertion latency becomes longer. Especially, the average latency under 100,000 partitions is 2.35x higher than it (2.35 seconds) under 10,000 partitions. Too many partitions incur lots of key-value pairs to be ingested into Redis cluster. The status information of each key-value pair is required to return to Aserv synchronously from Redis cluster to determine the successful ingestion. This procedure depends on the network response delay. Thus, our partition strategy can improve the insertion performance.

\textbf{Data reduction rate}. Although we use DAfilter to reduce data size, but the large number of partitions will also result in a poor data reduction. We use $MD/OD$ as the data reduction rate. $OD$ is the size of original data which is about 1,176 GB with 1,920 times of data collection and 1,064 GB with 1,738 times. $MD$ is the memory size consumed by Redis cluster. As shown in Figure \ref{fig:datareductionratio}, the number of partitions is negatively correlated with the data reduction rate, because Redis cluster needs to take more overhead to keep more key-value pairs, such as the extra overhead of data structure, etc. With the reduction of partitions, Aserv can achieve a satisfactory data reduction rate. For example, only 23 GB RAM per instance is consumed when the number of partitions is 10,000. It suggests less budget to pay for cloud resources. However, as the number of partitions comes to 100,000, the consumed memory is even greater than the size of original data causing DAfilter failure. Thus, less partitions in Aserv can reduce the memory consumption.

\begin{figure*}[t]
  \begin{minipage}[t]{0.33\linewidth} 
    \centering
    \includegraphics[width=2in]{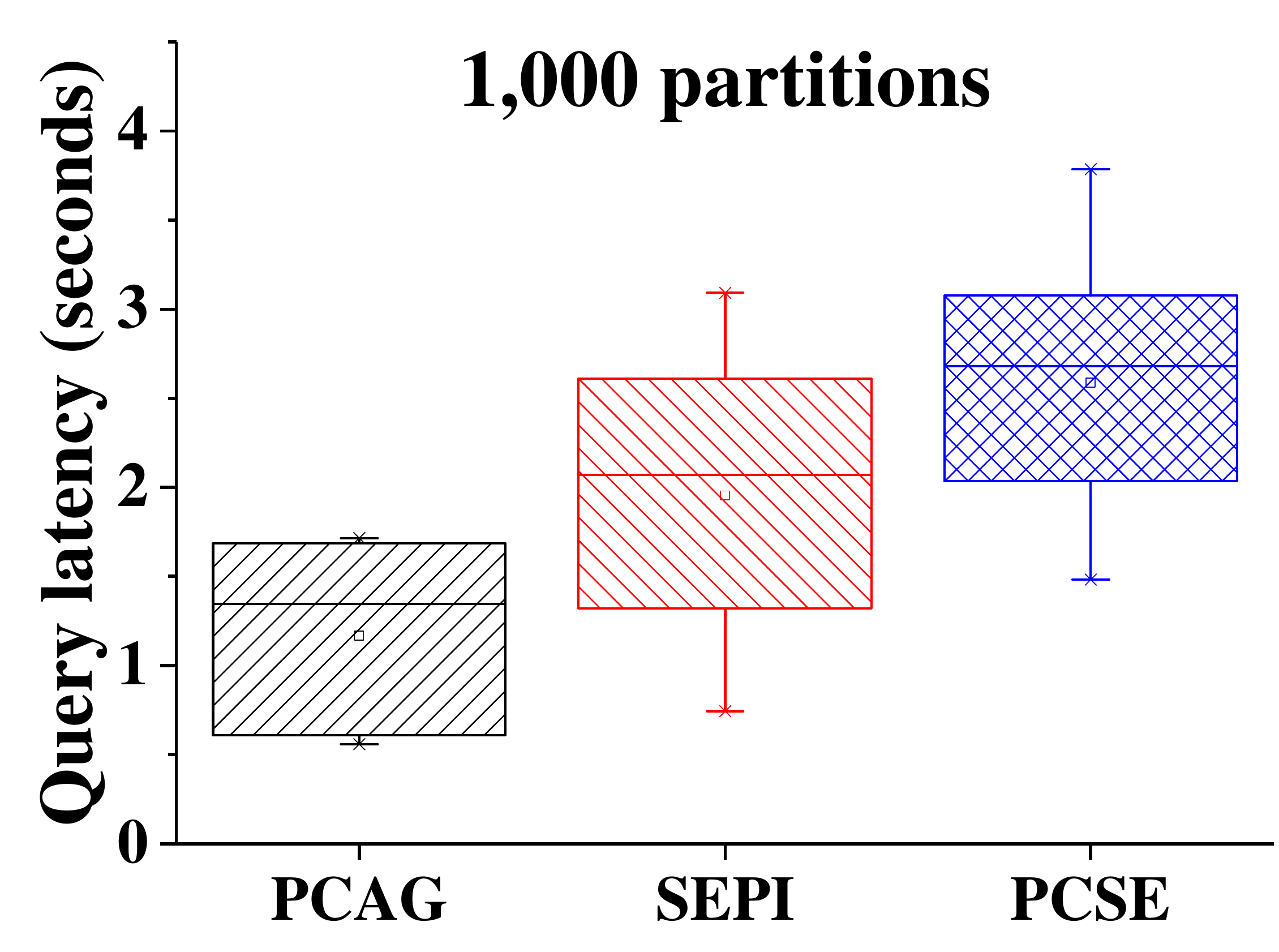}

  \end{minipage}%
  \begin{minipage}[t]{0.33\linewidth}
    \centering
    \includegraphics[width=2in]{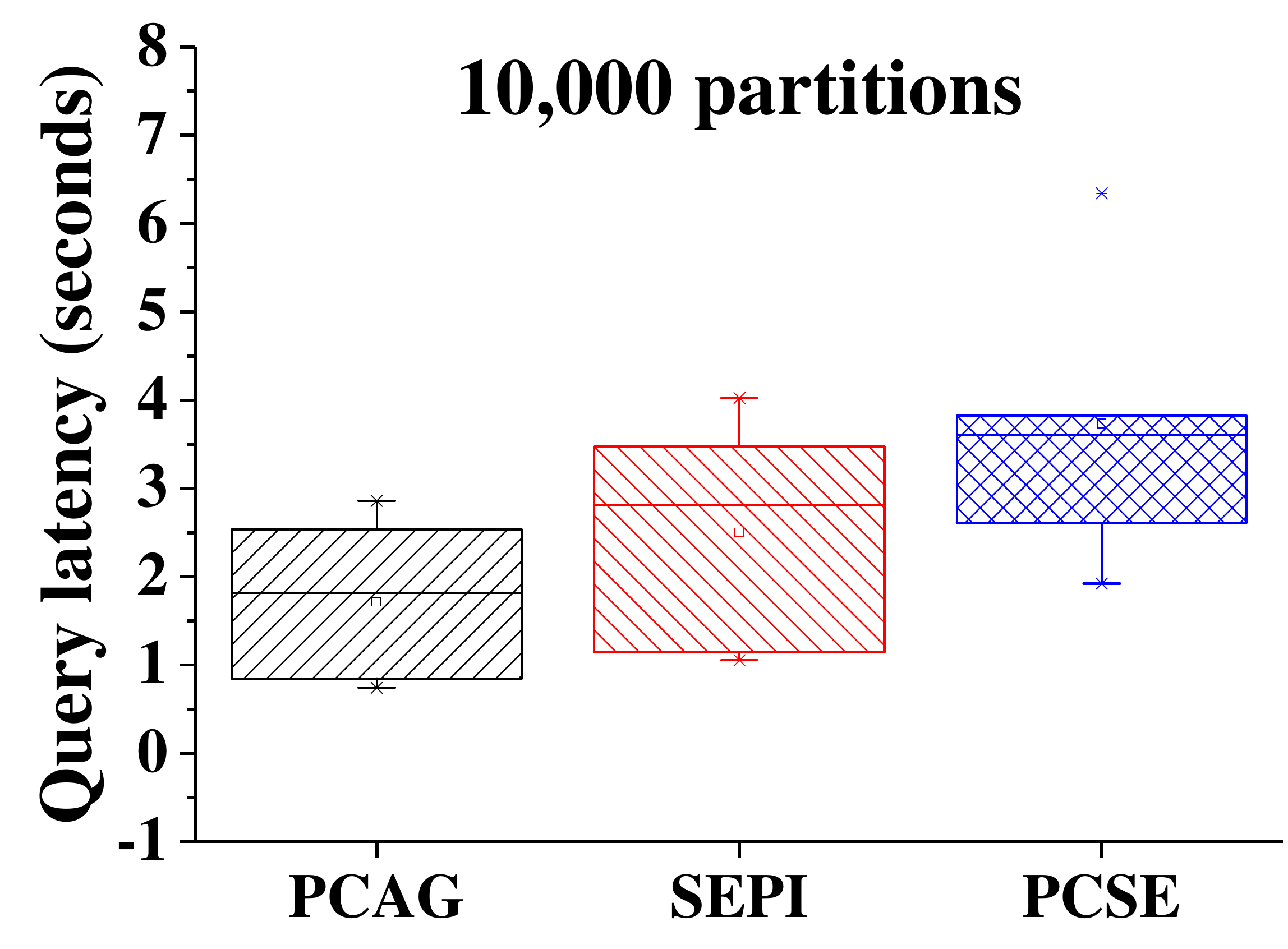}
  \end{minipage}
  \begin{minipage}[t]{0.33\linewidth}
    \centering
    \includegraphics[width=2in]{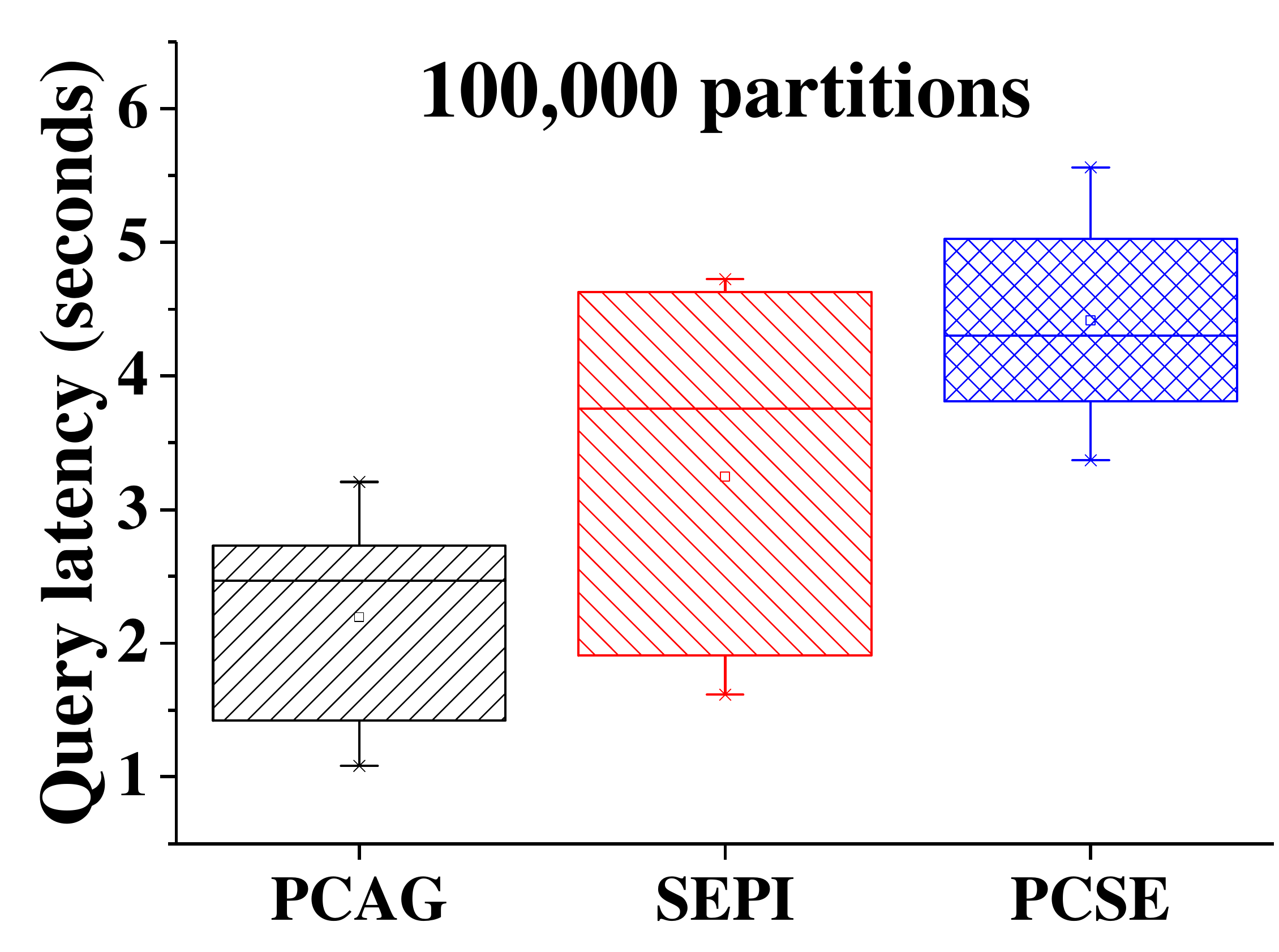}
  \end{minipage}
      \caption{Query latency for probing analysis}
    \label{fig:querylatencyforscientificeventcount}
\end{figure*}

\begin{figure*}[t]
  \begin{minipage}[t]{0.33\linewidth} 
    \centering
    \includegraphics[width=2in]{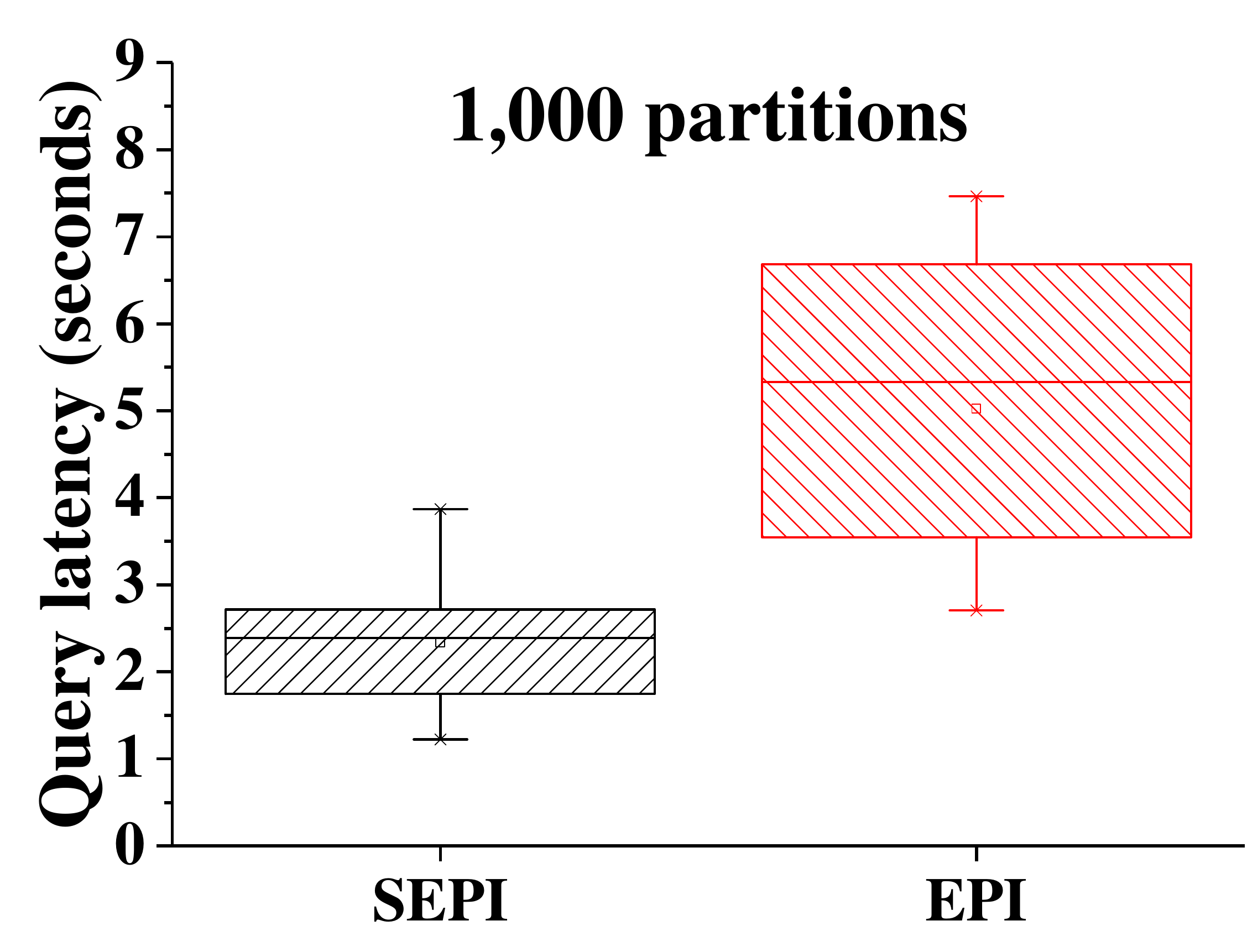}

  \end{minipage}%
  \begin{minipage}[t]{0.33\linewidth}
    \centering
    \includegraphics[width=2in]{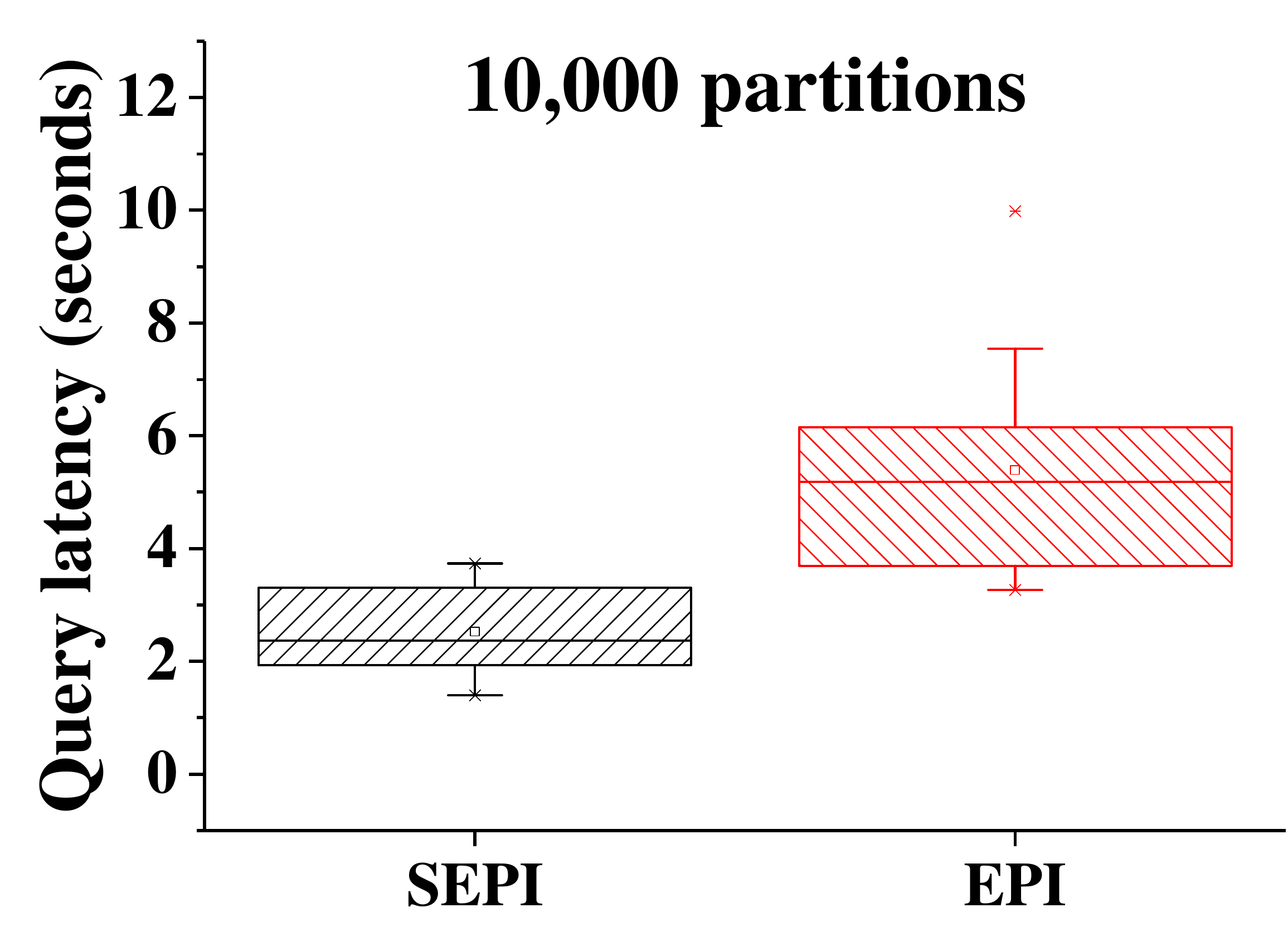}
  \end{minipage}
  \begin{minipage}[t]{0.33\linewidth}
    \centering
    \includegraphics[width=2in]{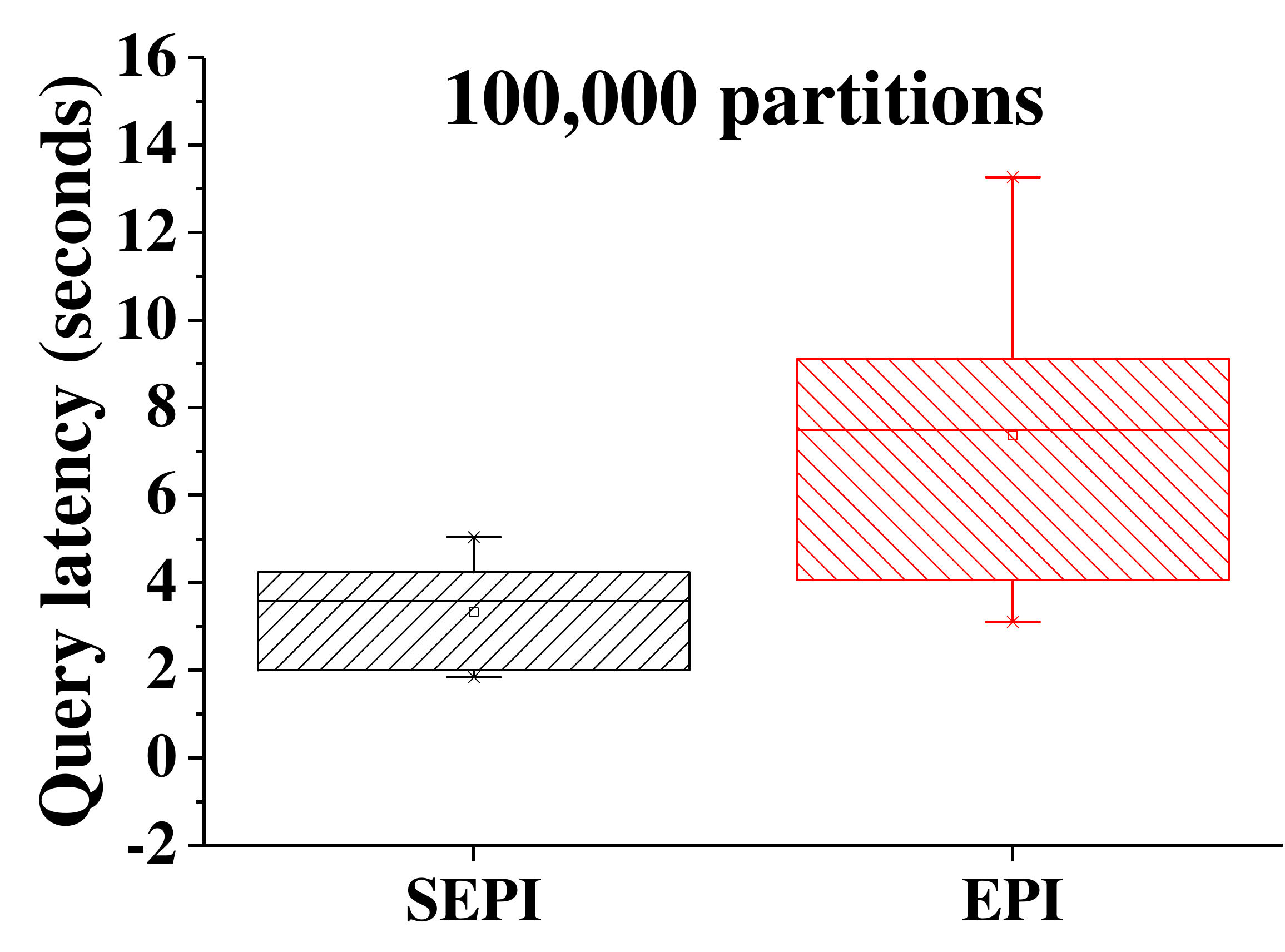}
  \end{minipage}
      \caption{Query latency for listing analysis}
    \label{fig:querylatencyforintervalquery}
\end{figure*}

\subsection{Query Performance Evaluation}\label{section:ClusterSize}

We refer to the LSST telescope discovery ability \cite{becla2008organizing} and assume that the accuracy of event detector is 0.5. Finally, we decide to simulate the generation of 200,000 scientific events one night. We randomly set 10 different sets of parameters for $region$ and $timeinterval$ operators to ensure that probing analysis and listing analysis can find about 50$\sim$5,000 scientific events. The result size is often concerned by scientists in interactive analysis. We evaluate the query latency on the maximum data size (i.e., 8 hour data for GWAC). It can represent the worst query performance in Aserv, because the data size processed by a query who is issued during the observation is less than the maximum data size.

\textbf{Probing analysis}. As shown in Figure \ref{fig:querylatencyforscientificeventcount}, we implement probing analysis by three different approaches and compare their performance. Probing analysis using PCAG and SEPI can get the approximate count. PCSE\footnote{Precise count of scientific events: first executing listing analysis and then counting the scientific events who really are inside the given search circle.} can get the precise count. At the three cases, the average query latency using PCAG is about 1.65 seconds which is 1.57x and 2.28x faster than it using SEPI and PCSE, respectively. It is obvious that scanning SEPI and parsing $region$ operator for every scientific event have more overhead than count by merging ICRs. In addition, as partitions becomes more, the query performance significantly reduces. For example, the query performance under 100,000 partitions reduces by 22\% compared with it under 10,000 partitions. The main reason is that it will incur more overhead to load more partitions and parse $region$ operator on them. However, it does not mean that the number of partitions should be as few as possible. At the extreme case, all objects are assigned into one partition, and it is hard to use PCAG to determine the approximate count.

\textbf{Query accuracy}. Although the reduction of partitions can improve the insertion and query performance, the query accuracy would reduce due to the failure of approximate region search. Actually, the query accuracy using PCAG and SEPI is the same due to the same way to parse $region$ operator so that we only use $count_{PCSE}/count_{PCAG}$ as the query accuracy. We evaluate the accuracy of 10 probing analytical queries, respectively and demonstrate the minimum accuracy as the final result in Figure \ref{fig:queryaccuracy}. It is within expectations that the actual query accuracy is 0.9 greater than the acceptable accuracy 0.8 under 10,000 partitions. Compared with 100,000 partitions, through 6.5\% accuracy loss we achieve the significant performance improvement.

\textbf{Listing analysis}. As shown in Figure \ref{fig:querylatencyforintervalquery}, we demonstrate the performance of listing analysis implemented by SEPI index and EPI index\cite{sfakianakis2013interval} in Redis cluster, respectively. Two different implementations parse the same $region$ operator, so that the necessary region search dose not lead to performance differences. The average query latency of three partition cases by SEPI is average 2.72 seconds that can meet the performance constraints. In addition, the query performance by SEPI is 2.22x higher than it by EPI. We consider that the key reason for the poor performance of EPI is two scans of key-value store and $distinct()$. Therefore, we capture time $T_s$ spent on two scans and find that it is 53.33\% of time spent on parsing EPI. In addition, $distinct()$ spends about 7\% of parsing time under 10,000 partitions. It can be explained that two scans and $distinct()$ have the high overhead. Fortunately, SEPI excludes one extra scan and $distinct()$. Moreover, the single scan of SEPI needs to load less key-value pairs than EPI because the SEPI's size is half of EPI. It leads that time spent on $scan()$ for SEPI only is 36.25\% of $T_s$.

\textbf{Stretching analysis}. We randomly select a scientific event which has been found by listing analysis. Furthermore, set $timeinterval(stime-\Delta t_1,etime+\Delta t_2)$ as the temporal range, where $stime$ and $etime$ are two endpoints of this scientific event and $\Delta t_1=\Delta t_2=20$ minutes being from the real demand of astronomers. We find that the average query latency is 0.34 second under 1,000 partitions, 0.2 second under 10,000 partitions and 0.19 second under 100,000 partitions, respectively. The query performance slightly reduces under less partitions due to more objects in each partition. The outstanding performance is because Aserv only needs to load one partition of data involving the corresponding object.
\subsection{Performance Constraint Evaluation}\label{section:performanceConstraintEvaluation}
In this section, we evaluate the accuracy of performance model when a cluster size $K$ is given. We set $K=19$ excluding one master node and use $acc_p=1-|T_e-T_a|/T_a \in [0,1]$ as the prediction accuracy, where $T_e$ is the estimated execution time and $T_a$ is the actual execution time. If the prediction accuracy is less than 0, we will set it into 0. We use the case of 10,000 partitions to take experiments, and all tests will take 1,920 times of data collection.

\textbf{Insertion latency prediction}. We build Aserv's cluster on two cloud instances to simulate the data generation of an observation unit. One instance is for the insertion component and another is for Redis cluster. The size of data being collected will be about equal to 1/19 of the total data size. Finally, we use the average insertion latency $f_p+f_s$ in Eq. (\ref{eq:insertionlatency}) as the estimated insertion latency $T_e$ of 19 nodes. We find that $T_e$ is 2.25 seconds which is less than $ct$. It suggests that Aserv's insertion latency can meet the performance constraint when $K=19$. In addition, we have solved $T_a=2.35$ seconds in Figure \ref{fig:Insertionlatency}. Therefore, the prediction accuracy is 0.96. It is explained that both the insertion component and Redis cluster have a good linear scalability.
\begin{table}[tbp]
\caption{Training data and predicted query latency}
 \centering
\begin{tabular}
{|m{2cm}<{\centering}|m{1.6cm}<{\centering}|m{1.6cm}<{\centering}|m{1.6cm}<{\centering}|}
\hline
\textbf{}&\textbf{Probing analysis}&\textbf{Listing analysis}&\textbf{Stretching analysis}\\
\hline
$f_r+f_q$&0.574 &1.07 &0.122 \\
\hline
$f_o(3)$&0.06 & 0.3&0.144 \\
\hline
$f_o(5)$&0.22 & 0.404&0.148 \\
\hline
$f_o(10)$&0.606 & 0.664&0.151 \\
\hline
$T_a$ (seconds)&1.72 & 2.52 & 0.202 \\
\hline
$T_e$ (seconds)&1.88 & 2.2 & 0.161 \\
\hline
$acc_p$&0.905& 0.873& 0.805\\
\hline
\end{tabular}
\label{tab:predictionresultsofquerylatency}
\end{table}

\textbf{Query latency prediction}. To evaluate the parallel time $f_r+f_q$, data distribution on one cloud instance needs to be simulated when $K=19$. Partition IDs are designed to be continuous. They naturally subject to uniform distribution, so that Redis cluster can easily place data evenly over the cluster. Therefore, we employ the modulus strategy for partition IDs (i.e., modulo 19) to select the corresponding partitions and ingest them into Redis cluster. We build Aserv's cluster on two cloud instances, one of which is employed to simulate the data generation of 19 observation units with the modulus strategy and another is for Redis cluster. It is noted that we do not take this way to evaluate the insertion latency, because 19 insertion processes in one instance may share the resources to impact the insertion performance. Only on the instance that Redis cluster resides in, we launch query engine and use the actual execution time as $f_r+f_q$.

 For the scale overhead, we build Aserv's cluster on $K'$ instances and also simulate data distribution of 19 observation units with the modulus strategy to evaluate $f_o(K')$. Firstly, we try to employ the modulus strategy to ensure that the number of partitions on $K'$ instances is close to $K'$ times of the number of partitions on one instance when $K=19$, so that we set $K' \in \{3,5,10\}$. For example, partition IDs are modulo 6 when $K'=3$. Then, we launch query engine on $K'$ instances and capture the actual query latency to solve $f_o(K')$ as the training data. Finally, we use Levenberg-Marquardt solver \cite{Mor1978The} to find $f_o$ that best fits the training data.

 For each analysis method, we also run them 10 times with different parameters and list the average results in Table \ref{tab:predictionresultsofquerylatency}.  When $K=19$, the estimated query latency solved by Eq. (\ref{eq:querylatency}) is also less than 15 seconds. It suggests that the performance of query engine can meet the performance constraint. Although only 3 points are used for training, the average prediction accuracy is 0.86, which is high enough to help scientists estimate Aserv's query performance and cloud resource setup. On the one hand, our queries implemented on Spark do not contain the complex communication pattern. On the other hand, our strategy is effective to directly capture the parallel time in Aserv.  Therefore, we can use a linear model to best fit a few training data.

\section{Related Work}\label{section:relatedWork}
Real-time and low latency scientific event analysis in fast sky survey has not been previously addressed. Work related to ours includes efforts to (1) real-time databases, (2) scientific databases, (3) low latency cloud data stores, and (4) scientific event indexing.

\textbf{Real-time databases}. Real-time databases have been studied since 1980s, and the key goal is to enable as many real-time transactions as possible to meet their respective time constraints\cite{Abbott1988Scheduling}. Real-time databases are more concern with timeliness, not system speed\cite{Stankovic1988Misconceptions}, due to a basis hypothesis that catastrophic consequences do not happen in the real world if a transaction is finished within the deadline. Hence, many of works focus on scheduling\cite{Abbott1988Scheduling,Haritsa1993Value} and transaction processing\cite{Lindstr2008Real,Kang2017Reducing}. Storing and processing all the data under periodic time constraint can avoid data loss and ensure temporal data consistency maybe enough for traditional real-time databases, but the transient feature of scientific event requires that the online query latency should be as low as possible. Only in this way, we can exploit the value of scientific data. Periodic survey cycle and the unpredictability of scientific event propose the new challenge for online big scientific data analysis.

\textbf{Scientific databases}. As big scientific data becomes more and more important in scientific discovery, many scientific databases are developed. Some are for specific science projects, such as SkyServer\cite{szalay2002sdss} for SDSS and Qserv\cite{Wang2011Qserv} for LSST, etc. Others are more versatile, such as SciDB\cite{Stonebraker2013SciDB} and SciServer\cite{webpage:sciserver}, etc. SciServer, as a major upgrade of SkyServer, can support a collaborative research environment for large-scale data-driven science. However, they are exploited to store the long-term historical data for complex analysis. It cannot meet the requirements of online analysis. There are works that focus on real-time challenge of scientific data, but most of them\cite{Kruger2005An,wan2016column} are for processing and storing original data, not analyzing scientific events. In terms of scientific events, \cite{becla2008organizing} only can work on the new scientific event signal to be rapidly cross-correlated with the huge historical catalogs through the rational data organization. Real-time wildfire monitoring\cite{Koubarakis2013Real} is essentially consistent with ours. The observed wildfire is also a kind of continuous scientific events. However, \cite{Koubarakis2013Real} does not support the typical analysis methods proposed by us. 

\textbf{Low latency cloud data stores}. These distributed systems, such as Redis cluster\cite{webpage:Redis}, Memcached\cite{webpage:Memcached} and MemSQL\cite{Skidanov2016A}, leverage main memory techniques to store and manage lots of data to enable the low latency access. The lack of the specific optimization mechanism for scientific event is their major difference with Aserv. These systems may pay more access time or budget to analyze scientific events on cloud since they only store the original catalogs and the scientific event signal sets. In addition, they do not give any approach to provide the guaranteed performance. This is not acceptable in our scenario. Thus, these systems could only be taken as the basic storage systems of Aserv. Some works leverage features of specific applications to propose the optimal approaches to reduce data access latency. DITIR\cite{cai2017ditir} is a storage system for high throughput trajectory insertion and real-time temporal range query. It uses B+ tree as the basic index to support temporal range query well. However, the current structure is not suitable for region search supported by Aserv.

\textbf{Scientific event indexing}. We formulize scientific event data as time-evolving data\cite{salzberg1999comparison}, which can be analyzed by segment tree\cite{Berg1997More}. In our scenario, the index tier must support the frequent write and low latency scan. The typical approaches are delayed commit and variants of segment tree. Delayed commit, such as MRST\cite{sfakianakis2013interval} and PISA\cite{Huang2016PISA}, accumulates the insertion and update operations and perform them together with a fixed cycle, but the index rebuild\cite{sfakianakis2013interval} or window generation\cite{Huang2016PISA} will cause the load latency when frequently writing. Variants of segment tree, such as SB-tree\cite{yang2003incremental}, Balanced-tree\cite{Moon2003Efficient}, transform segment tree as other structures (e.g., B tree or red-black tree), to support insertion and update. However, rotation and split still be involved, causing them to be not efficient. In summary, segment tree are not suitable for our scenario. Endpoint indexes are also employed for scientific event analysis. EPI\cite{sfakianakis2013interval} has the poor query performance. EIIQHBase\cite{zhou2016efficient} is presented to optimize EPI's query performance. However, it must in advance know two endpoints of scientific events to build the index structure. In our scenario, scientific events are unpredictable, so EIIQHBase is not feasible. SEPI learns from the advantages of EPI's structure to support frequent insertion and update and further improve the query performance.
\section{Conclusion}\label{section:summary}
In this paper, we propose three basic analysis methods for the fast sky survey and develop a distributed system Aserv to implement real-time and low latency scientific event analysis. The unnecessary cost is cut down to help us achieve an accuracy aware approach to improve the analysis performance. We trade off the query accuracy, the resource consumption (mainly memory and network), insertion and query latency by modeling scientific event and adjusting the number of partitions. Ultimately, achieve the overall balance of the large-scale scientific data analysis system. The specific optimization methods include DAfilter, EPgrid, SEPI index, PCAG and a performance model. Aserv can be downloaded from \url{https://github.com/yangchenwo/Aserv.git}.
\section{Acknowledgement}
This research was partially supported by the grants from National Key Research and Development Program of China (No. 2016YFB1000602, 2016YFB1000603); Natural Science Foundation of China (No. 91646203, 61532016, 61532010, 61379050, 61762082); Fundamental Research Funds for the Central Universities, Research Funds of Renmin University (No. 11XNL010); and Science and Technology Opening up Cooperation project of Henan Province (172106000077).

\bibliographystyle{IEEEtran}
\bibliography{IEEEexample}
\end{document}